\documentclass[sigconf,authorversion]{acmart} 
\AtBeginDocument{%
  \providecommand\BibTeX{{%
    \normalfont B\kern-0.5em{\scshape i\kern-0.25em b}\kern-0.8em\TeX}}}

 \copyrightyear{2024} 
\acmYear{2024} 
\setcopyright{acmlicensed}\acmConference[ISSAC '24]{International Symposium on Symbolic and Algebraic Computation}{July 16--19, 2024}{Raleigh, NC, USA}
\acmBooktitle{International Symposium on Symbolic and Algebraic Computation (ISSAC '24), July 16--19, 2024, Raleigh, NC, USA}
\acmDOI{10.1145/3666000.3669704}
\acmISBN{979-8-4007-0696-7/24/07}

\usepackage{algorithm}
\usepackage{algpseudocode}

\usepackage{enumitem}

\AtEndPreamble{%
\theoremstyle{acmdefinition}
\newtheorem{remark}[theorem]{Remark}}

\begin{document}

\title{On Minimal and Minimum Cylindrical Algebraic Decompositions}


\author{Lucas Michel}
\affiliation{%
  \institution{University of Liège}
  \country{Belgium}}
\email{lucas.michel@uliege.be}

\author{Pierre Mathonet}
\affiliation{%
  \institution{University of Liège}
  \country{Belgium}}
\email{p.mathonet@uliege.be}

\author{Naïm Zénaïdi}
\affiliation{%
  \institution{University of Liège}
  \country{Belgium}}
\email{nzenaidi@uliege.be}


\begin{abstract}
We consider cylindrical algebraic decompositions (CADs) as a tool for representing semi-algebraic subsets of $\mathbb{R}^n$. In this framework, a CAD $\mathscr{C}$ is adapted to a given set $S$ if $S$ is a union of cells of $\mathscr{C}$. 
Different algorithms computing an adapted CAD may produce different outputs, usually with redundant cell divisions. In this paper we analyse the possibility to remove the superfluous data. More precisely we consider the set CAD$(S)$ of CADs that are adapted to $S$, endowed with the refinement partial order and we study the existence of minimal and minimum elements in this poset.

We show that for every semi-algebraic set $S$ of $\mathbb{R}^n$ and every CAD $\mathscr{C}$ adapted to $S$, there is a minimal CAD adapted to $S$ and smaller (i.e. coarser) than or equal to $\mathscr{C}$.
Moreover, when $n=1$ or $n=2$, we strengthen this result by proving the existence of a minimum element in CAD$(S)$. Astonishingly for $n \geq 3$, there exist semi-algebraic sets whose associated poset of adapted CADs does not admit a minimum. We prove this result by providing explicit examples. 
We finally use a reduction relation on CAD$(S)$ to define an algorithm for the computation of minimal CADs. We conclude with a characterization of those semi-algebraic sets $S$ for which CAD$(S)$ has a minimum by means of confluence of the associated reduction system.
\end{abstract}

\ccsdesc[500]{Computing methodologies~Representation of mathematical objects}
\keywords{Semi-algebraic set; cylindrical algebraic decomposition; partially ordered set; minimal and minimum element; abstract reduction system}



\maketitle

\vspace{3cm}
\section{Introduction}

Cylindrical algebraic decomposition (CAD) is a fundamental tool in computer algebra and in real algebraic geometry. Since its introduction by Collins \cite{collins1975} in the framework of real quantifier elimination, the scope of application of CAD and the concept itself have evolved in various directions. Algorithms computing CADs are constantly refined \cite{LazardStyle23, McCallum98, Brown01, RegularChains} and are nowadays used as building blocks in various fields where symbolic manipulations of polynomial equalities and inequalities are relevant.

From a theoretical standpoint, the existence of a CAD algorithm may also be seen as a structure theorem for semi-algebraic sets not only giving topological information about them (see for example \cite{bochnaketal1998, schwartzsharir83,basu2007}) but also providing a constructive and explicit representation of these subsets enabling symbolic and effective reasoning. In this framework, a CAD is said to be adapted to (or represents) a semi-algebraic set $S \subset \mathbb{R}^n$ if $S$ is a union of cells from the CAD.  

In practice, for computing a CAD adapted to a semi-algebraic set $S$, one describes $S$ by polynomial inequalities (conjunction, disjunction, and negation) associated with polynomials in a finite set $\mathcal{F} \subset \mathbb{Z}[x_1,\ldots,x_n]$. The CAD algorithms usually take $\mathcal{F}$ as input and the output CAD is therefore adapted to any semi-algebraic set defined by means of $\mathcal{F}$. For this reason they usually contain unnecessary cells for the description of $S$. Moreover, these outputs also depend on the choice of $\mathcal{F}$ and on the considered algorithm. It is then natural to seek for a method for simplifying a given CAD, while retaining only the relevant information needed to describe $S$. 

More precisely, we study in detail the set CAD$(S)$ of CADs (with respect to a fixed variable ordering) that are adapted to $S$.  
This set is naturally endowed with the partial order $\preceq$ defined by refinement: for $\mathscr{C}, \mathscr{C}' \in \text{CAD}(S)$, we write $\mathscr{C} \preceq \mathscr{C}'$ if every cell of $\mathscr{C}$ is a union of cells of $\mathscr{C}'$. Here, the main problems\footnote{These problems were already mentioned in the same framework in \cite{wilson,locatelli}.} under consideration are the existence and the effective construction of minimal and minimum elements of the poset $\text{CAD}(S)$. 
These problems, theoretical in nature, lead us to analyse in detail the more practical question of simplifying CADs by merging cells (see Section \ref{sec:red}).
Finally, the minimum number of cells of these minimal CADs provide a tight lower bound on the number of cells that a CAD adapted to a semi-algebraic set must contain and could therefore be used to benchmark CAD algorithms (see for example \cite{NuCAD, MLCAD}).  

In Section \ref{sec:minimal}, we develop the first results on minimal CADs. We show in particular that for every semi-algebraic set $S$ of $\mathbb{R}^n$ and every CAD $\mathscr{C}$ adapted to $S$, there is a minimal CAD adapted to $S$ and smaller (i.e. coarser) than or equal to $\mathscr{C}$. This elementary result settles the question of the existence of minimal CADs adapted to $S$.  
Section \ref{sec:minimum} proposes a thorough investigation of minimum CADs. If $n=1$ or $n=2$, we prove the existence of a minimum element in CAD$(S)$ for every semi-algebraic $S$ of $\mathbb{R}^n$. We actually use results concerning minimal CADs obtained in Section \ref{sec:minimal} to show that such CADs must coincide. On the contrary, for $n \geq 3$, we show that there exist semi-algebraic sets whose associated poset of adapted CADs does not admit a minimum by providing explicit examples. 
In Section \ref{sec:red}, we study a reduction relation on the poset $\text{CAD}(S)$. We first use it to define an algorithm that constructs a minimal CAD adapted to a given semi-algebraic set $S$ and smaller than or equal to a given CAD $\mathscr{C}$. This algorithm may be used as a post-processing operation of any CAD algorithm. 
Finally, we employ this relation to characterize those semi-algebraic sets $S$ for which there exists a minimum adapted CAD. More precisely, we prove that $\text{CAD}(S)$ admits a minimum if and only if the associated reduction system is confluent.

\section{Background and notation}
We consider a positive integer $n \in \mathbb{N}^*$ and CADs of $\mathbb{R}^n$ with respect to a fixed variable ordering. Since the cells of such CADs are defined inductively from cells of CADs of $\mathbb{R}^k$ ($k \leq n$) which are indexed by $k$-tuples, it is convenient to make tuple notation more concise. We identify the $k$-tuple $I=(i_1,\ldots,i_k)\in\mathbb{N}^k$ with the corresponding word $i_1\ldots i_k$. We say that $I$ is odd (resp. even) if $i_k$ is odd (resp. even). We denote by $\varepsilon$ the empty tuple, which corresponds to the empty word. For $j\in \mathbb{N}$ we denote by $I:j$ the $k+1$ tuple $(i_1,\ldots,i_k,j)$. 

\begin{definition}[see \cite{Arnon,basu2007}]\label{def:cad}
    A cylindrical algebraic decomposition (CAD) of $\mathbb{R}^n$ is a sequence $\mathscr{C} = (\mathscr{C}_1,\ldots,\mathscr{C}_n)$ such that  
    for all $k \in \{1,\ldots,n\}$, the set  $$\mathscr{C}_k= \left\{C_{i_1 \cdots i_k} |  \forall j \in \{1,\ldots,k\}, i_j \in \{1,\ldots,2u_{i_1\cdots i_{j-1}} +1\}\right\}$$
     is a finite semi-algebraic partition of $\mathbb{R}^k$ defined inductively by the following data:
    \begin{enumerate}[leftmargin=0.35cm]
        \item[$\bullet$] there exists a natural number $u_\varepsilon \in \mathbb{N}$ and real algebraic numbers\footnote{\label{note:none}Potentially none if $u_\varepsilon = 0$.
        }
        $\xi_{2} < \xi_{4} < \ldots <\xi_{2u_\varepsilon}$
        that define exactly all cells of $\mathscr{C}_{1}$ by
        \[C_{2j} = \{\xi_{2j}\},\, (1 \leq j \leq u_\varepsilon),\quad
             C_{2j+1}= (\xi_{2j}, \xi_{2(j+1)}),\,(0 \leq j \leq u_\varepsilon)
         \]
         with the convention that $\xi_{0} = -\infty$ and $\xi_{2u_\varepsilon+ 2} = +\infty$;
        \item[$\bullet$] for each cell $C_I \in \mathscr{C}_k$ ($k < n$), there exists a natural number $u_I \in \mathbb{N}$ and semi-algebraic continuous functions\footnote{Potentially none. See also Proposition \ref{ex:trousers1} for explicit examples of such CADs.}
        $\xi_{I: 2 } < \xi_{I:4} < \ldots <\xi_{I:2 u_I} : C_I \to \mathbb{R}$
        that define exactly all cells of $\mathscr{C}_{k+1}$ by
        {\small\begin{align*}
             C_{I:2j} &= \{(\textbf{a}, b) \in C_I \times \mathbb{R} \; | \; b = \xi_{I:2j}(\textbf{a})\} , \quad (1 \leq j \leq u_I),\\
             C_{I:2j+1} &= \{(\textbf{a}, b) \in C_I \times \mathbb{R} \; |
\;\xi_{I:2j}(\textbf{a}) < b < \xi_{I:2(j+1)}(\textbf{a})\},\quad (0 \leq j \leq u_I)
         \end{align*}}
         with the convention that $\xi_{I:0} = -\infty$ and $\xi_{I:2u_I 
+ 2} = +\infty$.
    \end{enumerate}
     We say that the element $C_{I}$ of $\mathscr{C}_k$
is a cell of index $I$. If $I$ is odd (resp. even), we say that this cell is a sector (resp. a section). When the context is clear, we identify $\mathscr{C}$ with $\mathscr{C}_n$.
\end{definition}

We usually denote by $S$ a semi-algebraic set of $\mathbb{R}^n$ and we denote by $S^c$ its complement in $\mathbb{R}^n$.

\begin{definition}
  We say that a CAD $\mathscr{C}$ is adapted to $S$ if $S$ is a union of cells of $\mathscr{C}$. We denote by $\text{CAD} (S)$ the set of all CADs adapted to~$S$.
\end{definition}

\begin{remark}\label{rem:Collins}
    Let $p \in \mathbb{N}^*$. A direct consequence of Collins' seminal theorem \cite{collins1975} tells us that if $S_1,\ldots,S_p \subset \mathbb{R}^n$ are semi-algebraic sets, then there exists a CAD adapted to all of them simultaneously, i.e. $\text{CAD}(S_1) \cap \ldots \cap \text{CAD}(S_p)$ is never empty.
\end{remark}

Since CADs are partitions, we naturally use the refinement order defined on the set of partitions to compare CADs.  
\begin{definition}\label{def:order}
    Let  $\mathscr{C}$ and $\mathscr{C}'$ be two CADs of $\mathbb{R}^n$. We say that $\mathscr{C}'$ is a refinement of $\mathscr{C}$ if every cell of $\mathscr{C}$ is a union of cells of $\mathscr{C}'$. We also say that $\mathscr{C}$ is smaller than or equal to $\mathscr{C}'$ and write $\mathscr{C} \preceq \mathscr{C}'$.
    \end{definition}
    This definition is already used in \cite{BrownSimple}, where $\mathscr{C}$ is said to be simpler than $\mathscr{C}'$ when $\mathscr{C} \prec \mathscr{C}'$ and where algorithms to compute simple CADs with respect to this order are devised.
We also recall the definition of minimal and minimum elements. 
    \begin{definition}
        A minimal CAD adapted to $S$ is a minimal element of ($\text{CAD}(S),\preceq$). In other words, $\mathscr{C} \in \text{CAD}(S)$ is a minimal CAD adapted to $S$ if 
        \[\forall \mathscr{C}' \in \text{CAD}(S), (\mathscr{C}'\preceq \mathscr{C})  \implies (\mathscr{C}' =  \mathscr{C}).\] 
        A minimum CAD adapted to $S$ is a minimum element of ($\text{CAD}(S),\preceq$). Namely,  $\mathscr{C} \in \text{CAD}(S)$ is a minimum CAD adapted to $S$ if 
    \[\forall \mathscr{C}' \in \text{CAD}(S), \mathscr{C}'\succeq \mathscr{C}.\]     
        We say that $S$ admits a minimum CAD if there exists a minimum CAD adapted to $S$.
    \end{definition}

\section{Minimal CAD{\small s}}\label{sec:minimal}
In this section we show the existence of minimal CAD adapted to any semi-algebraic set $S$. Then, we show the equivalence between the uniqueness of a minimal CAD and the existence of a minimum CAD. Finally, we give uniqueness results concerning the set of sections of the last level of minimal CADs. 

\begin{proposition}\label{prop:existMin}
If $\mathscr{C}$ is a CAD adapted to $S$, then there exists a minimal CAD adapted to $S$ that is smaller than or equal to $\mathscr{C}$. 
\end{proposition}

\begin{proof}
The set  $\text{SSP}(\mathscr{C})$ of partitions that are strictly smaller than $\mathscr{C}$ is a finite poset, by definition. We can find a minimal element by inspection: if $\mathscr{C}$ is not a minimal element in $\text{CAD}(S)$, we can find $\mathscr{C}'$ in $\text{CAD}(S)$ such that $\mathscr{C}'\prec \mathscr{C}$, and iterate the procedure with $\mathscr{C}'$.  This will end in a finite number of steps because $\text{SSP}(\mathscr{C})$ is finite and strictly contains $\text{SSP}(\mathscr{C}')$. Furthermore, this procedure will obviously result in a minimal CAD with the required properties.
\end{proof}

\begin{remark}\label{rem:naive}
   The determination of a minimal element in $\text{CAD}(S)$ presented in the proof of Proposition \ref{prop:existMin}  is in general not efficient. If $\mathscr{C}$ contains $K$ cells, then the set SSP$(\mathscr{C})$ of partitions that are strictly smaller than $\mathscr{C}$ contains $B_K - 1$ elements, where $B_K$ is the $K^\text{th}$ Bell number. This procedure will be improved by Algorithm \ref{algo:Min}.
\end{remark}

\begin{proposition}\label{prop:uniqueMin}
    There is a unique minimal CAD adapted to $S$ if and only if there exist a minimum CAD adapted to $S$.
\end{proposition}
\begin{proof}
    Let $\mathscr{M}$ be the unique minimal CAD adapted to $S$. If $\mathscr{C} \in \text{CAD}(S)$, Proposition \ref{prop:existMin} gives us a minimal CAD adapted to $S$ and smaller than or equal to $\mathscr{C}$. By assumption, this minimal CAD is $\mathscr{M}$. This shows that $\mathscr{M}$ is smaller than or equal to any CAD adapted to $S$, hence is a minimum. The other implication is direct.
\end{proof}

In view of Proposition \ref{prop:uniqueMin} it is natural to further analyse properties of minimal CADs adapted to a set $S$ in order to decide if they must coincide. The following lemma provides a partial result in this direction concerning the last level of such minimal CADs. In order to have the same treatment for the general case and for the special case of dimension one, it is useful to consider the algebraic real numbers $\xi_{2j}$ that define a CAD of $\mathbb{R}$ as images of semi-algebraic functions defined on the one point set $\mathbb{R}^0\cong\{0\}$, which we still denote $\xi_{2j}$, that is, we set $\xi_{2j}=\xi_{2j}(0)$. It will also be convenient to consider the cells of a CAD $\mathscr{C}_1$ of $\mathbb{R}$ to be built above the set $\mathbb{R}^0$ endowed with the CAD $\mathscr{C}_0=\{C_{\varepsilon}\}$ with $C_\varepsilon=\{0\}$.
For every $x\in\mathbb{R}^{n-1}$, we define $S_x= \{y \in \mathbb{R} \; |\; (x,y) \in S\}\subset \mathbb{R}$ and denote by $\partial S_x$ its boundary in $\mathbb{R}$.

\begin{proposition}\label{rem:top-caract}
    Consider $\mathscr{C} \in \text{CAD}(S)$, $x\in\mathbb{R}^{n-1}$ and $C_I \in \mathscr{C}_{n-1}$ such that $x \in C_I$. Then $\partial S_x\subset \{\xi_{I:2j}(x) \in \mathbb{R} \; | \; j \in \{1,\ldots,u_I\}\}$. If $\mathscr{C}$ is minimal, these two sets are equal.
\end{proposition}
\begin{proof}
    Consider $y\in\partial S_x$ and suppose that $y$ is not in $\{\xi_{I:2j}(x) \in \mathbb{R} \; | \; j \in \{1,\ldots,u_I\}\}$. Since the point $(x,y)$ is not in a section of $\mathscr{C}$, it must be in a sector $C_{I:2j+1}$ of $\mathscr{C}$ for a $j \in \{0,\ldots,u_I\}$. Hence, the real number $y$ belongs to the open interval $U = (\xi_{I:2j}(x), \xi_{I:2(j+1)}(x))$. Since $y$ is in $\partial S_x$, there exists $z \in U \cap S_x$ and $z' \in U \cap (S_x)^c$. This means that both points $(x,z)$ of $S$ and $(x,z')$ of $S^c$ belong to $C_{I:2(j+1)}$. This is a contradiction because $\mathscr{C}$ is adapted to $S$.

    We now consider a minimal $\mathscr{C} \in \text{CAD}(S)$ and prove the other inclusion by contradiction. 
    Suppose that there exists $j \in \{1,\ldots,u_I\}$ such that $\xi_{I:2j}(x)$ is not in $\partial S_x$. Then $\xi_{I:2j}(x)$ lies in the interior of $S_x$ or in the interior of $(S_x)^c$, and there exists $\varepsilon > 0$ such that the open interval $(\xi_{I:2j}(x) - \varepsilon, \xi_{I:2j}(x) + \varepsilon)$ is a subset of $S_x$ or  $(S_x)^c$. This implies that there exists $\eta > 0$ such that the points  $P_\alpha = (x,\xi_{I:2j}(x)+ \alpha \eta), \alpha \in \{-1,0,1\}$ belong simultaneously to $C_{I:2j + \alpha} \cap S$ or $C_{I:2j + \alpha} \cap S^c$.  Since $\mathscr{C}$ is adapted to $S$, the three cells  $C_{I:2j-1}, C_{I:2j},$ $C_{I:2j+1}$ are simultaneously subsets of $S$ or of $S^c$. Thus, the tuple 
    $\widetilde{\mathscr{C}} =(\mathscr{C}_1,\ldots, \mathscr{C}_{n-1}, \widetilde{\mathscr{C}}_n)$
    with 
    \begin{align*}
        \widetilde{\mathscr{C}}_n = \mathscr{C} _n \setminus \{C_{I:2j-1},C_{I:2j},C_{I:2j+1}\} \cup \{C_{I:2j-1} \cup C_{I:2j} \cup C_{I:2j+1}\}
    \end{align*} is a CAD adapted to $S$ and strictly smaller than $\mathscr{C}$. This is a contradiction with the minimality of $\mathscr{C}$.
\end{proof}

The previous result asserts that the union of sections of the last level of minimal CADs adapted to $S$ only depends on $S$. We obtain immediately the following result, that will be used in the next section.
\begin{corollary}\label{prop:outil}
    Let $\mathscr{C}$ and  $\mathscr{C}'$ be two minimal CADs adapted to $S \subset \mathbb{R}^n$ and $C_{I} \in \mathscr{C}_{n-1}, C'_{I'} \in \mathscr{C}'_{n-1}$. If the intersection $C_I \cap C'_{I'}$ is not empty then $u_I = u'_{I'}$ and for all $j \in \{1,\ldots,u_I\}$, the restrictions of the functions $\xi_{I:2j}$ and  $\xi'_{I':2j}$ to $C_I \cap C'_{I'}$ are equal.
\end{corollary}

\section{Minimum CAD{\small s}}\label{sec:minimum}
In this section we investigate the existence of a minimum element in $\text{CAD}(S)$. It happens that CADs have properties in low dimensions that are false in higher dimensions (see for instance \cite{lazard2010} or more recently \cite{locatelli}). In dimensions 1 and 2, we obtain an existence result for every $S$. However, in dimension $n \geq 3$, we exhibit examples of semi-algebraic sets admitting several distinct minimal adapted CADs and therefore no minimum adapted CAD. We first describe explicitly an elementary example in $\mathbb{R}^3$ defined only by means of linear constraints. 
We then use it to produce subsets of $\mathbb{R}^n$ ($n \geq 4$) with no minimum adapted CAD.

\begin{theorem}\label{thrm:existenceMinimum}
For every semi-algebraic set $S$ of $\;\mathbb{R}$ or of $\;\mathbb{R}^2$, the poset $\text{CAD}(S)$ has a minimum.
\end{theorem}
\begin{proof}
In view of Proposition \ref{prop:uniqueMin}, it is sufficient to show that if $\mathscr{C}$ and $\mathscr{C}'$ are two minimal elements in $\text{CAD}(S)$, then $\mathscr{C} = \mathscr{C}'$. We use the notation of Definition \ref{def:cad} for $\mathscr{C}$ and $\mathscr{C}'$ with obvious adjustments. 

Suppose first that the ambient space is $\mathbb{R}$. The projection of any cell of $\mathscr{C}$ or $\mathscr{C}'$ coincides  with $C_{\varepsilon}$, as explained before Proposition \ref{rem:top-caract}. This result directly implies that $\mathscr{C} = \mathscr{C}'$. 

Suppose now that the ambient space is $\mathbb{R}^2$. We first show that $\mathscr{C}_1 = \mathscr{C}'_1$.  Assume for the sake of contradiction that there exists $i \in \{1,\ldots,u_\varepsilon\}$ such that
     $\xi_{2i} \notin \{\xi'_{2j} \; | \; j \in \{1,\ldots,u_\varepsilon'\}\}.$ We will prove that we can merge the corresponding cells of the cylinders over $C_{2i-1}, C_{2i}$ and $C_{2i+1}$ to build a CAD $\widetilde{\mathscr{C}}\in \text{CAD}(S)$ that is strictly smaller than $\mathscr{C}$, which is a contradiction.  We first show that the semi-algebraic continuous functions that define the sections of the cylinders over $C_{2i-1}, C_{2i}$ and $C_{2i+1}$ can be glued together to define cells over the union $C_{2i-1}\cup C_{2i}\cup C_{2i+1}$. Since we have $\xi'_{2}<\cdots<\xi_{2u_\varepsilon'}$,  there exists $j \in \{1,\ldots,u_\varepsilon'+1\}$ such that
     $\xi'_{2(j-1)} < \xi_{2i} < \xi'_{2j}.$ By definition, $\xi_{2i}$ is in the closure of $C_{2i-1}$, $C_{2i}$ and $C_{2i+1}$, so the intersections $C_{2i-1} \cap C'_{2j-1},
         C_{2i} \cap C'_{2j-1}$ and $C_{2i+1} \cap C'_{2j-1}$ are not empty. Then we can use Corollary \ref{prop:outil} to obtain not only
     \[u'_{2j-1} = u_{2i-1}= u_{2i}= u_{2i+1},\]
     but also
     \begin{align*}
         \xi'_{2j-1:2k}(x) = \begin{cases}
             \xi_{2i-1 : 2k}(x) &\text{ if } x \in C_{2i-1} \cap C'_{2j-1},\\
             \xi_{2i : 2k}(x) &\text{ if } x \in C_{2i} \cap C'_{2j-1},\\
             \xi_{2i+1 : 2k}(x) &\text{ if } x \in C_{2i+1} \cap C'_{2j-1}.
         \end{cases}
     \end{align*}
     for every $k \in \{1,\ldots,u'_{2j-1}\}$. This implies that for any such $k$ the function $\widetilde{\xi}_{2i-1:2k}$ defined on $C_{2i-1} \cup C_{2i} \cup C_{2i+1}$ by
     \[ \widetilde{\xi}_{2i-1:2k}(x) = \begin{cases}
         \xi_{2i-1 :2k}(x) &\text{ if } x \in C_{2i-1},\\
         \xi_{2i : 2k}(x) &\text{ if } x \in C_{2i},\\
         \xi_{2i+1 : 2k}(x) &\text{ if } x \in C_{2i+1},
     \end{cases}\]
     is continuous, since it coincides with a continuous function in the neighbourhood of every point in its domain. It is also semi-algebraic since its graph is the union of the graphs of $\xi_{2i-1:2k}, \xi_{2i:2k}$ and $\xi_{2i+1:2k}$, which are all semi-algebraic by definition.
     Finally, for all $l \in \{1,\ldots,2u'_{2j-1}+1\}$, the three cells $C_{2i-1:l},$ $C_{2i:l}$ and $C_{2i+1:l}$ have respectively a non-empty intersection with the cell $C'_{2j-1:l}$. Using the fact that $\mathscr{C}$ and $\mathscr{C}'$ are adapted to $S$, these four cells are simultaneously either a subset of $S$ or of $S^c$.
     Therefore, the couple $\widetilde{\mathscr{C}} = (\widetilde{\mathscr{C}}_1,\widetilde{\mathscr{C}}_2)$
defined by
    \begin{align*}
        \widetilde{\mathscr{C}}_1 = \mathscr{C}_1 \setminus& \{C_{2i-1},C_{2i},C_{2i+1}\} \cup \{C_{2i-1} \cup C_{2i} \cup C_{2i+1}\},\\
        \widetilde{\mathscr{C}}_2 = \Big(\mathscr{C} _2 \setminus& \big\{C_{2i-1 : l},C_{2i:l},C_{2i+1:l} \; | \; l \in \{1,\ldots,2u'_{2j-1}+1\}\big\}\Big) \\
        &\cup \Big\{C_{2i-1:l} \cup C_{2i:l} \cup C_{2i+1:l}\; | \; l \in \{1,\ldots,2u'_{2j-1}\}\Big\}.
    \end{align*}
    is a CAD that is adapted to $S$ and strictly smaller than $\mathscr{C}$ as requested.
    This shows that we have the inclusion
    \[\big\{\xi_{2i} \; | \; i \in \{1,\ldots,u_{\varepsilon}\} \big\} \subset \big\{\xi'_{2j} \; | \; j \in \{1,\ldots,u_\varepsilon'\}\big\}.\]
    Since the other inclusion follows by symmetry, these finite sets are equal, implying that $u_{\varepsilon} = u_\varepsilon'$ and that $\xi_{2k} = \xi'_{2k}$ for all $k \in \{1,\ldots,u_{\varepsilon}\}.$ In particular, we have $\mathscr{C}_1 = \mathscr{C}_1'$. Using again Corollary \ref{prop:outil} above each cell of $\mathscr{C}_1 = \mathscr{C}_1'$, we obtain that $\mathscr{C}_2 = \mathscr{C}_2'$. This means that $\mathscr{C} = \mathscr{C}'$, as announced.
\end{proof}

\begin{remark}
    Proposition \ref{rem:top-caract} provides a topological characterization of the minimum of CAD($S$) when $S \subset \mathbb{R}$. Its set of sections is exactly the boundary of $S$.
\end{remark}

We now present an example of a semi-algebraic set of $\mathbb{R}^3$ admitting several distinct adapted minimal CADs, hence no minimum adapted CAD.

\begin{definition}
    We define the Trousers $\mathbb{T}$ as the semi-algebraic set given by
     \begin{align*}
        \mathbb{T} = &\big\{(x,y,z) \in \mathbb{R}^3 \;|\; (x \leq 0 \lor y \leq 0)\land z = 0\big\}\\
              &\cup \big\{(x,y,z) \in \mathbb{R}^3 \;|\; x > 0 \land y > 0 \land z = -x/2\big\}.
     \end{align*}
\end{definition}

Notice that $\mathbb{T}$ is built only with linear polynomials.
Observe that in Figure~\ref{fig:trousers} the dashed half-line is not included in $\mathbb{T}$, while the thick half-line is included in $\mathbb{T}$.

\begin{figure}[H]
    \center
    \includegraphics[scale=0.35]{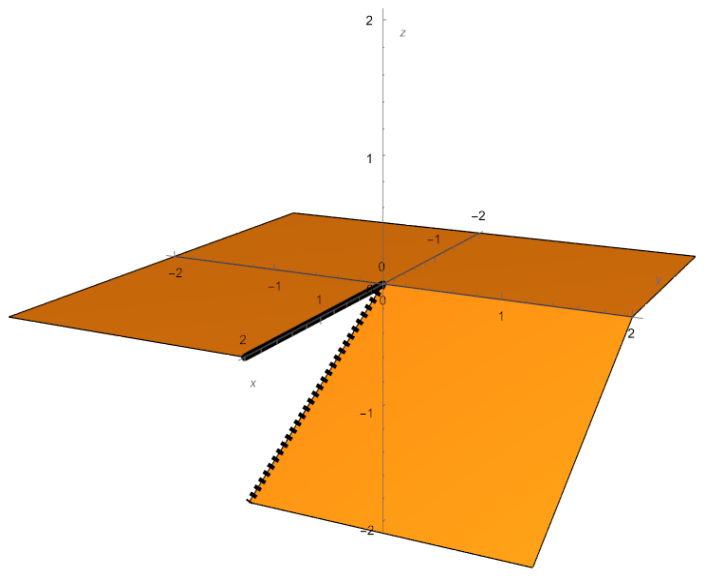}
    \caption{The Trousers $\mathbb{T}$}
    \label{fig:trousers}
\end{figure}

\begin{proposition}\label{ex:trousers1}
    There is no minimum CAD adapted to the trousers $\mathbb{T}$.
\end{proposition}
\begin{proof}
    We describe two distinct minimal CADs $\mathscr{C}$ and $\mathscr{C}'$ which are adapted to $\mathbb{T}$ and conclude by Proposition \ref{prop:uniqueMin}.
    For the first one, the partition $\mathscr{C}_1$ contains the unique cell $C_1 = \mathbb{R}$ while $\mathscr{C}_2$ consists in the three cells $C_{11} = C_1 \times (-\infty, 0), C_{12} = C_1 \times \{0\}, C_{13} = C_1 \times (0,+\infty).$ The partition $\mathscr{C}_3$ is built with the nine cells
    \begin{align*}
        C_{1j1} &= C_{1j} \times (-\infty,0),  &C_{131} \equiv (x,y) \in C_{13}\land z < -\frac{x}{2} \chi(x,y), \\
        C_{1j2} &= C_{1j} \times \{0\}, &C_{132} \equiv (x,y) \in C_{13}\land z = -\frac{x}{2} \chi(x,y),\\
        C_{1j3} &= C_{1j} \times (0,+\infty), &C_{133} \equiv (x,y) \in C_{13}\land z > -\frac{x}{2} \chi(x,y),
    \end{align*}
    for $j \in \{1,2\}$, where $\chi$ is the indicator function of the set $\{(x,y) \in \mathbb{R}^2 \; | \; x > 0 \land y > 0\}$. In other words, for each cell $C_{11}, C_{12}$ and $C_{13}$, there is a unique section defined by the function $x \mapsto -\frac{x}{2} \chi(x,y)$.
    For the second one, the partition $\mathscr{C}'_1$ contains the three cells $C'_1 = (-\infty,0)$, $C'_2 = \{0\}$, $C'_3 = (0,+\infty)$ while $\mathscr{C}'_2$ consists in the cells $C'_{11} = C'_1 \times \mathbb{R}, C'_{21} = C'_2 \times \mathbb{R}, C'_{31} = C'_3 \times (-\infty, 0), C'_{32} = C'_3 \times \{0\}, C'_{33} = C'_3 \times (0,+\infty).$ The partition $\mathscr{C}'_3$ is built with the fifteen cells
    \begin{align*}
        C'_{ik_i1} &= C'_{ik_i} \times (-\infty,0), &C'_{331} \equiv (x,y) \in C'_{33}\land z < -\frac{x}{2}, \\
                C'_{ik_i2} &= C'_{ik_i} \times \{0\},& C'_{332}  \equiv (x,y) \in C'_{33}\land z = -\frac{x}{2}\\
                C'_{ik_i3} &= C'_{ik_i} \times (0,+\infty), &C'_{333} \equiv (x,y) \in C'_{33}\land z > -\frac{x}{2},
        \end{align*}
    for $i \in \{1,2,3\}, k_1 = k_2 = 1, k_3 \in \{1,2\}.$ 
    \begin{figure}[H]
        \center
        \includegraphics[scale=0.15]{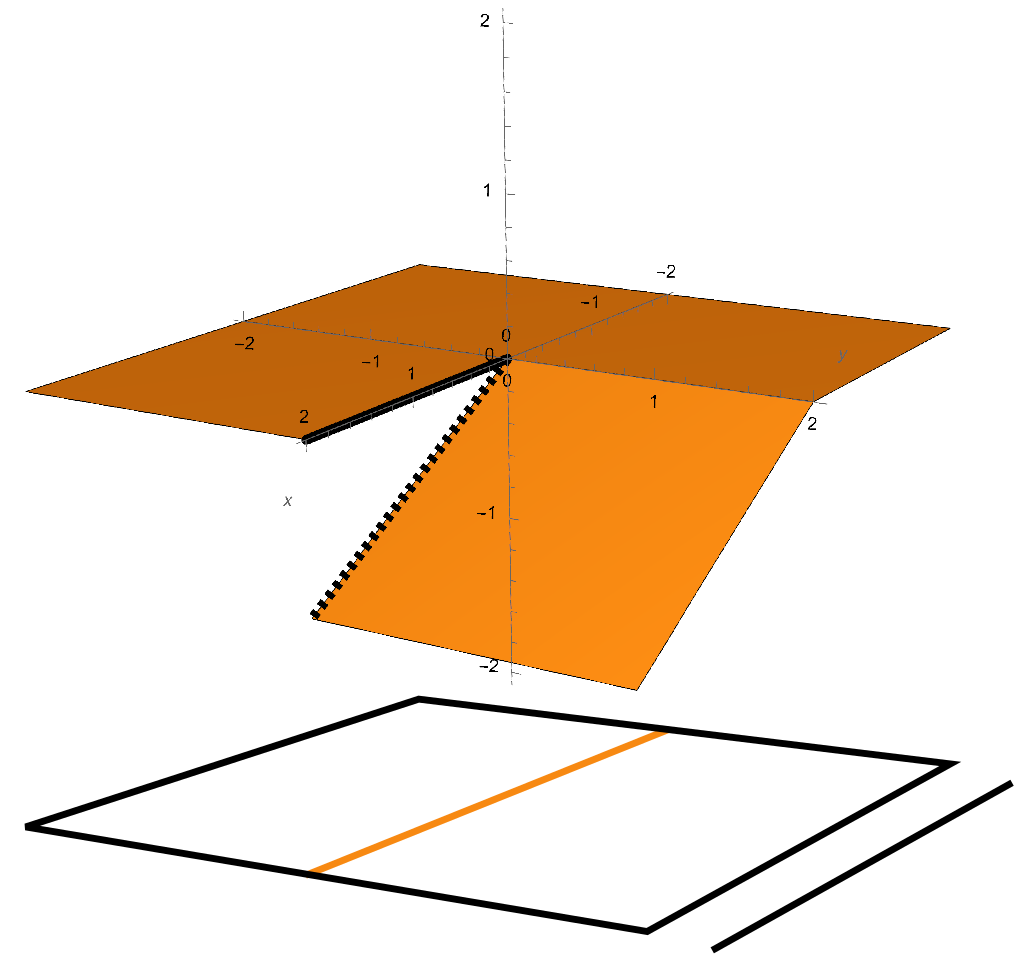}~
        \includegraphics[scale=0.15]{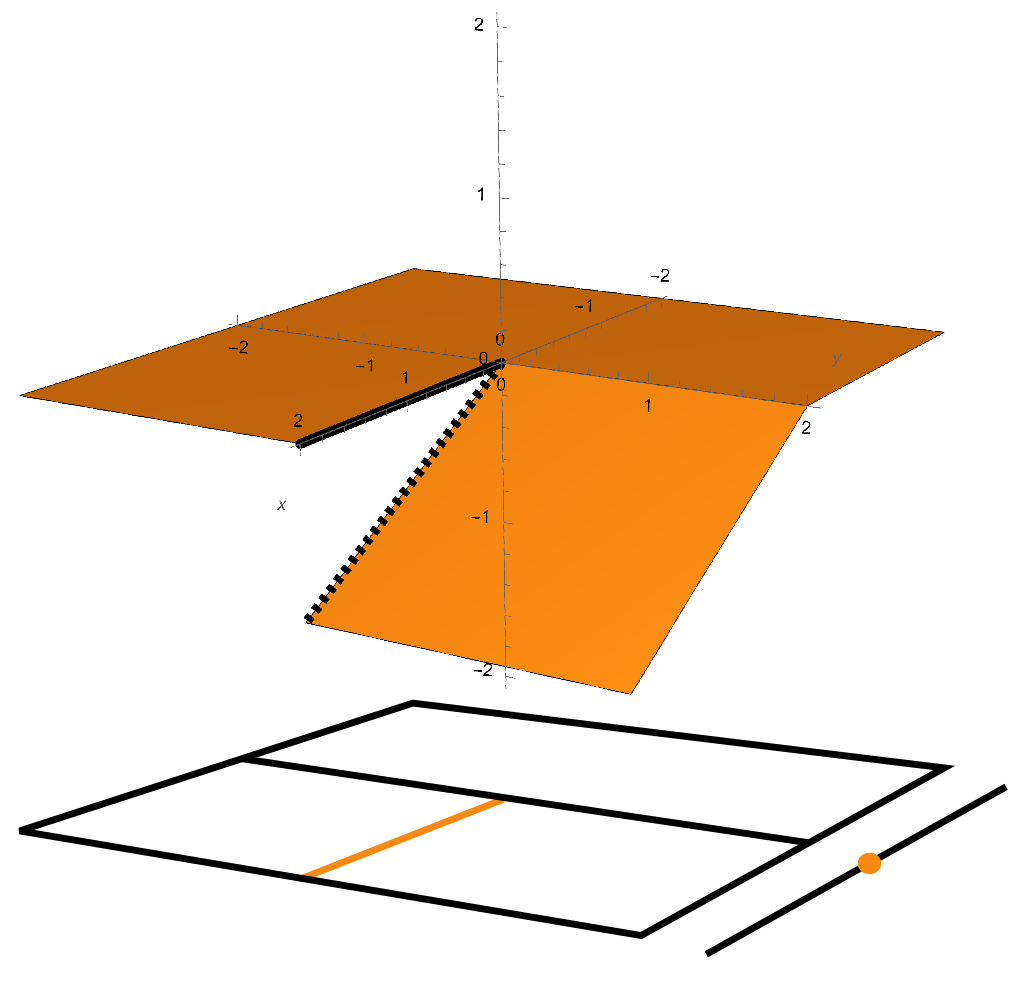}
        \caption{The CADs $\mathscr{C}$ and $\mathscr{C}'$}
        \label{fig:trousersCADs}
    \end{figure}
    These CADs $\mathscr{C}$ and $\mathscr{C}'$, which are obviously distinct, are adapted to $\mathbb{T}$ because
    \begin{align*}
        \mathbb{T} &= C_{112} \cup C_{122} \cup C_{132}\\
        &= C'_{112} \cup C'_{212} \cup C'_{312}\cup C'_{322} \cup C'_{332}.
    \end{align*}
    A careful but direct inspection then shows that these CADs are indeed minimal.
\end{proof}

\begin{remark}\label{rem:trousersHard}
    In order to show the minimality of $\mathscr{C}$ and $\mathscr{C}'$ in the previous proof by considering all the partitions that are strictly smaller than $\mathscr{C}$ and $\mathscr{C}'$, one is led to investigate the sets $\text{SSP}(\mathscr{C})$ and $\text{SSP}(\mathscr{C}')$ which contain respectively $B_9 - 1 =  21\,146$ and $B_{15} - 1 =  1\, 382\, 958\,
    544$ elements, as discussed in Remark \ref{rem:naive}. Taking into account the properties of CADs (e.g. the topology of cells and the cylindrical character of the partition) drastically reduces the number of partitions to be taken into consideration. Hence, the brute force approach is doomed to be inefficient.
    This observation leads us to study in the next section the notion of CAD reduction. This concept will enable us to analyse the order on CAD$(S)$, and for a given CAD $\mathscr{C}$, quickly remove many candidates from SSP($\mathscr{C}$) that will never give rise to a CAD. We will then easily justify the minimality of the CADs given in the proof of Proposition \ref{ex:trousers1} (see Example \ref{ex:trousers2}).
\end{remark}

We now extend Proposition \ref{ex:trousers1} in higher dimension.

\begin{corollary}
    The set $\mathbb{T}_n$ defined for $n \geq 3$ by
    \begin{align*}
        \mathbb{T}_n = &\big\{(x_1,\ldots,x_n) \in \mathbb{R}^n \; | \; (x_1 \leq 0 \lor x_2 \leq 0) \land (x_3 = 0)\big\}\\
            &\cup \big\{(x_1,\ldots,x_n) \in \mathbb{R}^3 \; | \; x_1 > 0 \land x_2 > 0 \land x_3 = -x_1/2\big \}
    \end{align*}
    admits no minimum adapted CAD.
\end{corollary}
\begin{proof}
    By the very definition of $\mathbb{T}_n$, we have $(x_1,x_2, x_3) \in \mathbb{T}$ if and only if for all $x_4,\ldots,x_n \in \mathbb{R}$, $(x_1,\ldots,x_n) \in \mathbb{T}_n$. It follows that if $\mathcal{D} = (\mathcal{D}_1, \mathcal{D}_2, \mathcal{D}_3)$ is a CAD adapted to $\mathbb{T}$, then the tuple $\widetilde{\mathcal{D}} = ( \mathcal{D}_1, \mathcal{D}_2, \mathcal{D}_3, \mathcal{D}_4, \ldots, \mathcal{D}_n)$, where
    \[\mathcal{D}_k = \big\{D \times \mathbb{R} \; | \; D \in \mathcal{D}_{k-1}\big\}\]
    for $k \geq 4$, is a CAD adapted to $\mathbb{T}_n$. Considering again the CADs $\mathscr{C}$ and $\mathscr{C}'$ adapted to $\mathbb{T}$ described in the proof of Proposition \ref{ex:trousers1}, a careful inspection shows that the corresponding CADs $\widetilde{\mathscr{C}}$ and $\widetilde{\mathscr{C}}'$ are two distinct minimal CADs adapted to $\mathbb{T}_n$.  
\end{proof}

\begin{remark}
    Here again, the proof of the minimality of $\widetilde{\mathscr{C}}$ and $\widetilde{\mathscr{C}}'$ will be greatly simplified by using the tools developed in Section \ref{sec:red}.
\end{remark}

We now exhibit two more examples of closed sets, one bounded and one unbounded, that do not admit minimum adapted CADs. These constructions can be readily extended to higher dimensions.
\begin{example}
    Neither of the semi-algebraic sets
    \begin{align*}
        \mathcal{B} = &[-1,1] \times [-1,1] \times [0,1] \\
        &\cup \big\{(x,y,z) \in \mathbb{R}^3 \; | \; y = 0, 0 \leq x \leq 1, 1 \leq z \leq x+1\big\},\\
        \mathcal{U} = &\big\{(x,y,z) \in \mathbb{R}^3 | \; (x \leq 0 \lor y \leq 0)  \land z \leq 0\big\}\\
            & \cup \big\{(x,y,z) \in \mathbb{R}^3 |\; x \geq 0 \land y \geq 0 \land x+yz \leq 0 \land z \leq 0\big\}
    \end{align*}
    admits a minimum CAD. We can indeed describe two distinct minimal CADs adapted to $\mathcal{U}$ using the same  $\mathscr{C}$ and $\mathscr{C}'$ defined for the trousers in the proof of Proposition \ref{ex:trousers1}, but replacing all occurrences of $-x/2$ with $-x/y$. The construction of two minimal CADs adapted to $\mathcal{B}$ follows the same lines with suitable adaptations.  It is therefore omitted.
    \begin{figure}[H]
        \center
        \includegraphics[scale=0.27]{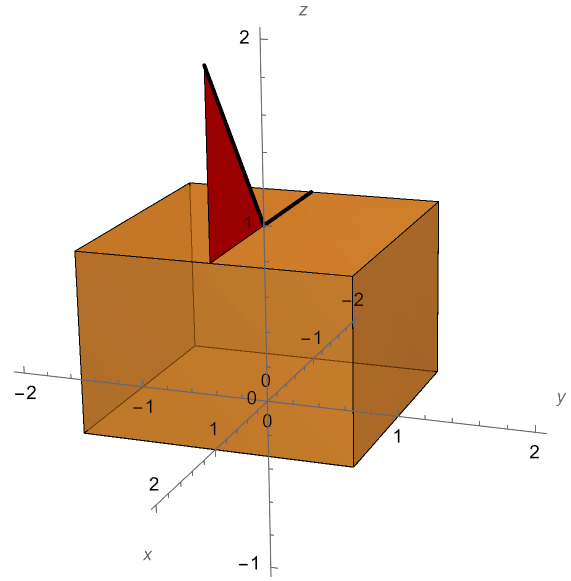}~
        \includegraphics[scale=0.24]{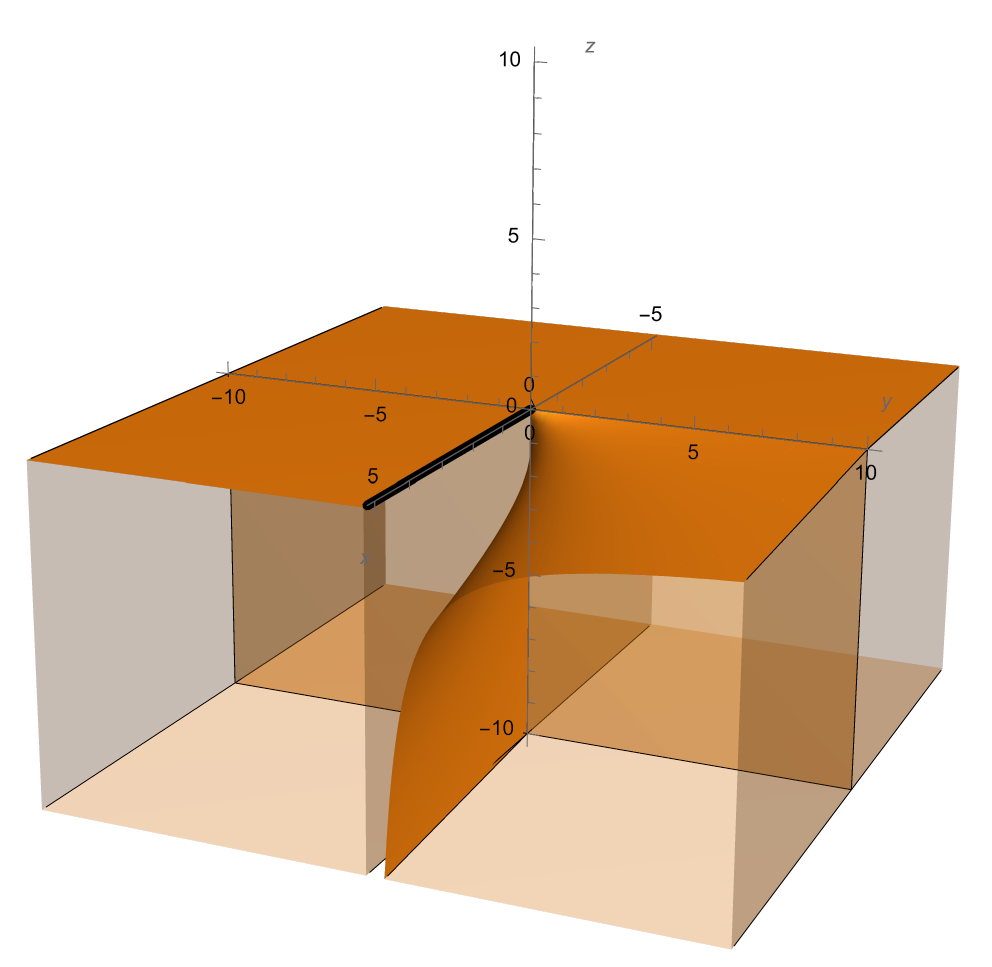}
        \caption{Semi-algebraic sets with two distinct minimal CADs}
    \end{figure}
\end{example}

\section{CAD reductions}\label{sec:red}

As already mentioned in Remark \ref{rem:trousersHard}, the determination of CADs that are adapted to a set $S$ and smaller than a given CAD $\mathscr{C}$ by direct inspection is certainly inefficient and already cumbersome even for low dimensional examples (see Section \ref{sec:minimum}).
The notion of CAD reduction that is developed in this section offers a better grasp on the partially ordered set ($\text{CAD}(S), \preceq$). It enables us to introduce an algorithm for simplifying the computation of minimal CADs. We apply this algorithm to the CADs introduced in the proof of Proposition \ref{ex:trousers1}, and provide a short proof of their minimality. 

\subsection{Tree structure of CAD}

The tree associated with a given CAD $\mathscr{C}$ can be seen as a simplified version of $\mathscr{C}$, which only encodes its combinatorial structure. Roughly speaking, this is done by considering only the cell indices of the considered CAD and a function that encodes the inclusion of the cells in the set $S$ under consideration. 
We first introduce the general notion of a CAD tree.

\begin{definition}\label{def:CADTREE}
    A CAD tree of depth $n$ is a pair $\mathcal{T} = (T,L)$ where 
        \begin{itemize}
            \item $T$ is a labelled rooted odd-ary tree whose nodes are labelled by tuples. The root is labelled by the empty tuple $\varepsilon$ and if $I$ is the label of a node, then either $I$ is an $n$-tuple and the node is a leaf, or there exists $u_I \in \mathbb{N}$ such that the node has exactly $2u_I + 1$ sons labelled by $I : j$ with  $j \in \{1,\ldots,2u_I + 1\}$. As we continue, we identify a node with its label, thus seeing $T$ as a subset of $\bigcup_{k=0}^n(\mathbb{N}^*)^k$ endowed with the prefix relation.
            \item $L$ is a map defined recursively on $T$ from the leaves to the root by
                \[L(I) \begin{cases}
                    \in \{0,1\} &\text{ if $I$ is a leaf of $T$,}\\
                    = \big(L(I:1), \ldots, L(I:2u_I+1)\big) &\text{ otherwise.}
                \end{cases} \]   
        \end{itemize}    
\end{definition}

\begin{remark}
    The redundancy in the definition above is useful for writing \eqref{eqn:L} below easily. However, it is clear that the function $L$ is completely defined by its values on the leaves of $T$.
\end{remark}

\begin{definition}\label{def:treeCAD}
    For every $\mathscr{C} \in \text{CAD}(S)$, the CAD tree associated with $(\mathscr{C}, S)$ is the CAD tree of depth $n$ where the nodes of $T$ are the indices of the cells of $\mathscr{C}$ and where the function $L$ is defined on the leaves of $T$ by $L(I) = 1$ if $C_I \subset S$ and by $L(I) = 0$ if $C_I \subset S^c$. We denote it by Tree$(\mathscr{C} ,S)$ or Tree$(\mathscr{C})$ for short.
\end{definition}

\begin{example}\label{ex:diskMotiv}
    Let $S$ be the closed disk centered at the origin and of radius one. 
    Consider the CAD $\mathscr{C} = (\mathscr{C}_1,\mathscr{C}_2)$ adapted to $S$ where the five cells of $\mathscr{C}_1$ are defined by $C_1 = (-\infty,-1)$, $C_2 = \{-1\},$ $C_3 = (-1,1)$, $C_4 = \{1\}$, $C_5 = (1,+\infty)$ and where the cylinders above these cells are split by exactly the two functions defined on $C_2$ and $C_4$ by $0$ and by the two functions defined on $C_3$ by $\sqrt{1-x^2}$ and $-\sqrt{1-x^2}$.
    We also consider another CAD $\mathscr{C}'$ adapted to $S$. It is obtained by splitting the cell $C_3$ into three cells $ (-1,0), \{0\}, (0,1)$. The functions used to cut cylinders of this new CAD are the same as those for $\mathscr{C}$, subject to appropriate restrictions. 
    Figure \ref{fig:disk} represents these CADs and their associated CAD trees. 
    A leaf $I$ such that $L(I)$ is $0$ (resp. $1$) is represented in red (resp. green). Moreover, we shorten the labelling of the nodes in an obvious manner.
    \begin{figure}
        \includegraphics[width=0.48\linewidth]{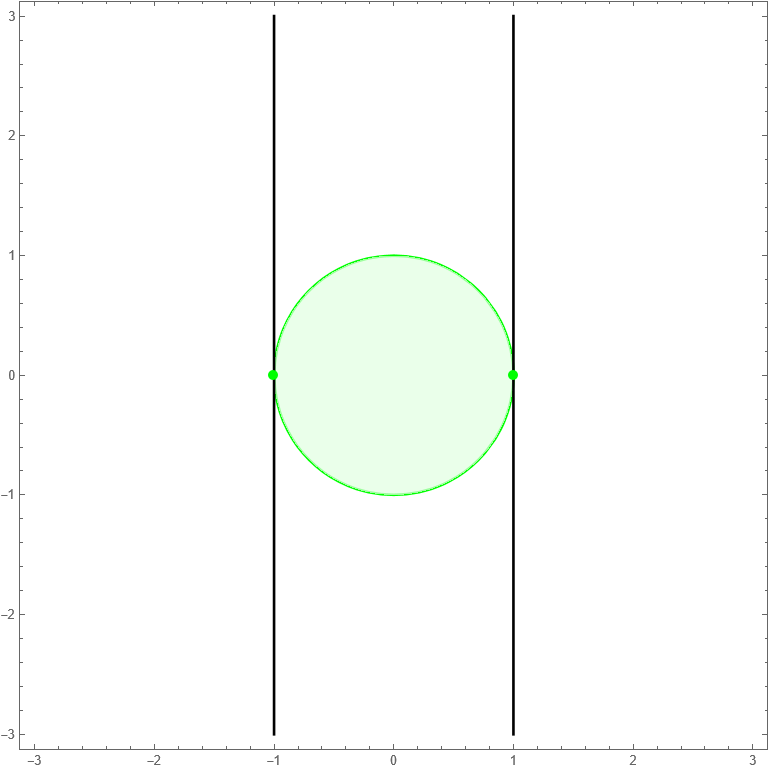}
        \includegraphics[width=0.48\linewidth]{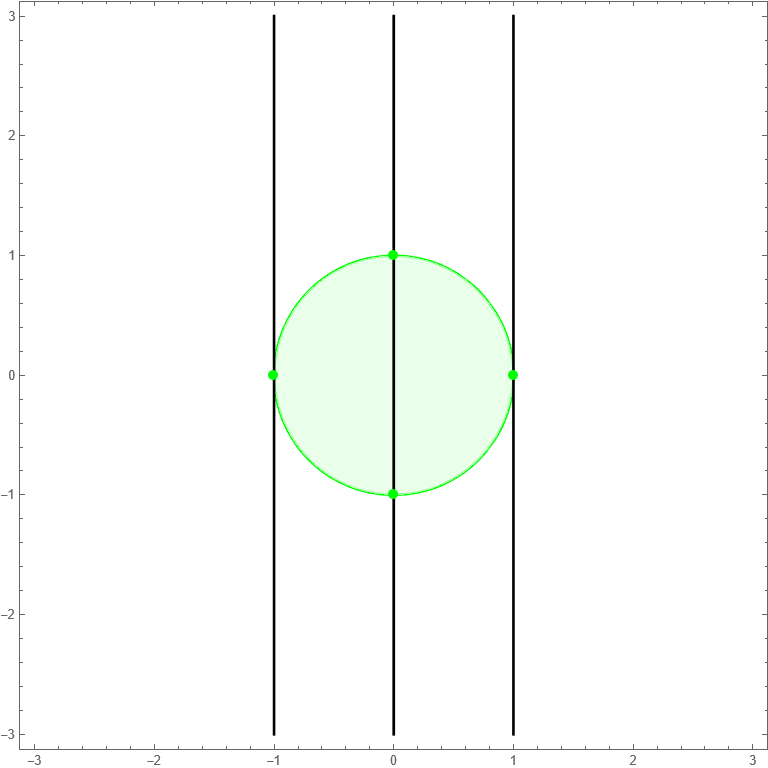}
        \includegraphics[width=\linewidth]{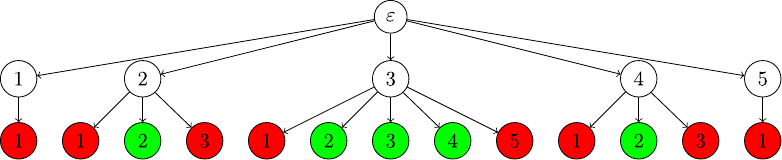}
        \includegraphics[width=\linewidth]{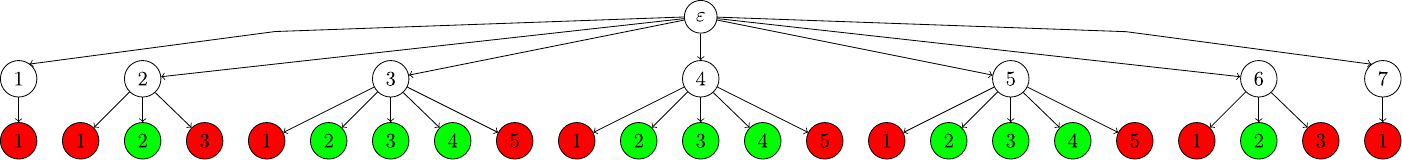}
        \caption{The CADs $\mathscr{C}$ and $\mathscr{C}'$ and their associated CAD trees}
        \label{fig:disk}
    \end{figure}
\end{example}

\subsection{Reductions}
The concepts of reductions, which we now introduce, can already be seen in Example \ref{ex:diskMotiv}. Indeed, the CAD $\mathscr{C}$ can be obtained by merging corresponding cells of $\mathscr{C}'$ above the third, fourth and fifth cylinders. Furthermore, the possibility to merge these cells of $\mathscr{C}'$ implies the following on the associated CAD trees:
\begin{itemize}
    \item the subtree with prefix 3 of Tree$(\mathscr{C})$ is identical to the subtrees with prefixes 3,4 and 5 of Tree$(\mathscr{C}')$;
    \item the subtrees with prefixes 4 and 5 of Tree$(\mathscr{C})$ are respectively identical to the subtrees with prefixes 6 and 7 of Tree$(\mathscr{C}')$. 
\end{itemize}
In this example, we will say that there is a reduction from $\text{Tree}(\mathscr{C}')$ to $\text{Tree}(\mathscr{C})$. It is obtained by removing subtrees with prefixes 4 and 5 and properly relabelling the other nodes (by subtracting $2$ to the 1-prefixes of subtrees with prefixes 6 and 7). In order perform these operations in the general case, we make use of the $k^\text{th}$ prefix map
    \[p_{k} \colon \bigcup_{l=0}^{+\infty}(\mathbb{N}^*)^l \to \bigcup_{l=0}^{+\infty}(\mathbb{N}^*)^l \colon (i_1,\ldots,i_l) \mapsto \begin{cases}
        (i_1,\ldots,i_l)&\text{ if $l < k$},\\
        (i_1,\ldots,i_k)&\text{ if $l\geq k$}.
       \end{cases}
    \]
When we merge a section with (even) index $A$ with two surrounding sectors in a level of a CAD, we have to merge the subtrees corresponding to these cells and relabel the whole tree accordingly. This is the purpose of the functions we introduce in the next definition.
\begin{definition}\label{def:auxpsi} 
    For every even tuple $A \in (\mathbb{N}^*)^k$ ($k \geq 1$),  we define 
    {\small\begin{align*}
        \psi_A \colon \bigcup_{l=0}^{+\infty}(\mathbb{N}^*)^l \to \bigcup_{l=0}^{+\infty}(\mathbb{N}^*)^l \colon I \mapsto \begin{cases}
            I - e_k \text{ if } p_k(I) = A,\\
            I - 2 e_k \text{ if } \exists m \in \mathbb{N}^* :  p_k(I) = A+me_k,\\
            I \text{ otherwise},
        \end{cases}
    \end{align*}}
    where $e_k$ is the $k^\text{th}$ unit vector of $\mathbb{R}^l$ for $l \geq k$.
\end{definition}

\begin{definition}\label{def:CADtreeRed}
    Consider a CAD tree $\mathcal{T} = (T,L)$ and an even node $A \in T \cap (\mathbb{N}^*)^k$. We say that $\psi_A$ induces a reduction rule on $\mathcal{T}$ if we have
    \begin{equation}\label{eqn:L}
        L(A - e_k) = L(A) = L(A + e_k).
    \end{equation}
    Then the induced reduction rule is denoted by $\Psi_{A}$ and the reduced CAD tree is given by $\mathcal{T}' =  (\psi_{A}(T), L')$ where $L'$ is defined on the leaves of $\psi_A(T)$ by $L \circ \psi_{A}^{-1}$. We write\footnote{Here we make use of the terminology on Abstract Reduction Systems from \cite{TRaAT}.} $\mathcal{T} \to \mathcal{T'}$ or $\mathcal{T}' \leftarrow \mathcal{T}$.
 \end{definition}

 The definition of $\psi_A$ and Condition \eqref{eqn:L} ensure that $\mathcal{T}'$ is a well-defined CAD tree. The existence of a reduction rule $\Psi_{A}$ from Tree$(\mathscr{C})$ characterizes the fact that the cylinders above $C_{A - e_k}, C_A$ and $C_{A + e_k}$ of $\mathscr{C}$ contain the same number of cells and the corresponding ones are all in $S$ or all in $S^c$. We will say that $\Psi_A$ induces a CAD reduction if we can ensure that the unions of cells prescribed by $\Psi_A$ give rise to a CAD.

\begin{definition}\label{def:CADRed}
    Let $\mathscr{C}$ be in $\text{CAD}(S)$. 
    A CAD tree reduction rule $\Psi_{A}$ from $\text{Tree}(\mathscr{C})$ to $\mathcal{T}'$ lifts to a CAD$(S)$ reduction rule $\Phi_A$ defined on $\mathscr{C}$ if 
    \[\mathscr{C}' = \left\{\bigcup_{I \in \text{Tree}(\mathscr{C}) \; : \; \psi_A(I) = I'} C_{I} \; \Big| \; I' \text{ leaf of } T'\right\}\]
    is a CAD adapted to $S$. In this case, the reduced CAD is given by $\mathscr{C}'$ and we write $\mathscr{C} \to \mathscr{C}'$ or $\mathscr{C}' \leftarrow \mathscr{C}$.
\end{definition}

The existence of the CAD tree reduction rule $\Psi_{A}$ is a necessary condition for the existence of the corresponding CAD reduction rule $\Phi_{A}$. It provides an algorithmically-friendly way to generate the candidates for CAD($S$) reductions. When this necessary condition is satisfied, the existence of $\Phi_A$ depends on the geometry of the CAD, in particular, but not only, on the adjacency of the cells.

\begin{example}
    In Example \ref{ex:diskMotiv}, the reduction rule $\Psi_4$ from Tree$(\mathscr{C}')$ to Tree$(\mathscr{C})$ exists and lifts to a reduction rule $\Phi_4$ from $\mathscr{C}'$ to $\mathscr{C}$. The reduction rule $\Phi_4$ consists in merging together the corresponding cells from the third, fourth and fifth cylinders of $\mathscr{C}'$, giving the CAD $\mathscr{C}$. 
    An example of a CAD tree reduction rule that does not lift to CAD reduction rule is studied in Example \ref{ex:trousers2}.
\end{example}
 
The following theorem enables us to study the poset $\text{CAD}(S)$ via CAD reductions.
We denote by $\stackrel{*}{\leftarrow}$ the reflexive-transitive closure of the relation $\leftarrow$ on $\text{CAD}(S)$.
\begin{theorem}\label{prop:lien-red-ordre}
    Let $\mathscr{C}, \mathscr{C} ' \in \text{CAD}(S)$. We have $\mathscr{C}' \stackrel{*}{\leftarrow} \mathscr{C}$ if and only if $\mathscr{C}' \preceq \mathscr{C}.$
\end{theorem}
\begin{proof} By definition of CAD reductions, if $\mathscr{C}' \stackrel{*}{\leftarrow} \mathscr{C}$, then all the cells of $\mathscr{C}'$ are unions of cells of $\mathscr{C}$. Hence, we have $\mathscr{C}' \preceq \mathscr{C}.$  

Suppose now that $\mathscr{C}' \prec \mathscr{C}$. Since the set 
    $\{\mathscr{B} \in \text{CAD}(S) \; | \; \mathscr{C}' \prec \mathscr{B} \prec \mathscr{C}\}$ is finite, it is sufficient to prove that there exists $\mathscr{D} \in \text{CAD}(S)$ such that $\mathscr{C}' \leftarrow \mathscr{D}\preceq \mathscr{C}$. We will build such a CAD $\mathscr{D}$ by adding a suitable section of $\mathscr{C}$ to the sections of $\mathscr{C}'$.
    
    Specifically, we first notice that $\mathscr{C}'_k \preceq \mathscr{C}_k$ for all $k \in \{0,\ldots, n-1\}$ and $\mathscr{C}'_n \prec \mathscr{C}_n$. It follows that there exists $ p \in \{1,\ldots,n\}$ such that $
        \mathscr{C}'_k = \mathscr{C}_k$ if $k < p$ and $\mathscr{C}_{k}' \prec \mathscr{C}_k$ if $k \geq p$. By definition, there exists a cell $C'_{J:s} \in \mathscr{C}'_p$ which is a non trivial union of cells of $\mathscr{C}_p$, i.e. we can write $C'_{J:s} = \bigcup_{a \in A} C_{I_a : r_a}$ where $A$ is a finite set with at least two elements. Both sides of the equality project onto the same cell $C'_J = C_J \subset \mathbb{R}^{p-1}$. Hence, we have $I_a = J$ for all $a \in A$. 
        
        If $x \in C'_J$, we observe that the set $\{y \in \mathbb{R} \; | \; (x,y) \in C'_{J:s}\}$ cannot be reduced to a singleton. It is then an open interval, which means that $C'_{J:s}$ is a sector. Furthermore, there exists $a\in A$ such that $r_a$ is even and is thus the index of a section of $\mathscr{C}_p$. Thus, we can write $s = 2j+1, r_a = 2l$ for $j \in \{0,\ldots, u'_J\}, l \in \{1,\ldots,u_J\}$.
        By definition, the sector $C'_{J:2j+1}$ is defined using the functions $\xi'_{J:2j}$ and $\xi'_{J:2(j+1)}$. Since $C_{J:2l} \subset C'_{J:2j+1}$, we have $\xi'_{J:2j} < \xi_{J:2l} < \xi'_{J:2(j+1)}$ over $C'_J$.

        Then, we construct the CAD $\mathscr{D} = (\mathscr{D}_1,\ldots,\mathscr{D}_n)$ explicitly by
    \begin{itemize}
        \item setting $\mathscr{D}_k = \mathscr{C}_k$ if $k < p$;
        \item replacing the sector $C'_{J:2j+1}$ in $\mathscr{C}'_p$ by the three cells obtained by slicing it with $\xi_{2l}$ to obtain $\mathscr{D}_p$;
        \item slicing accordingly the cells of $\mathscr{C}'_k$ for $k > p$ above $C'_{J:2j+1}$ to obtain $\mathscr{D}_k$.
    \end{itemize} 
    By construction, there exists a CAD$(S)$ reduction rule $\Phi_{J:2(j+1)}$ from $\mathscr{D}$ to $\mathscr{C}'$.  
    We finally show by induction that $\mathscr{D}_k \preceq \mathscr{C}_k$ for $k \in \{p,\ldots,n\}$.
\end{proof}

\begin{example}\label{ex:hasseDisk}
    We consider the same disk $S$ and CADs $\mathscr{C}$ and $\mathscr{C}'$ as in Example \ref{ex:diskMotiv}. 
    We define the CAD $\mathscr{C}''$ depicted in the top of Figure~\ref{fig:hasse} by slicing the cylinders above $C'_3, C'_4$ and $C'_5$ by means of the function $f$ defined on the interval $(-1,1)$ by $f(x) =  1 - (2(x^2-1))^{-1}$.
    Figure~\ref{fig:hasse} shows the Hasse diagram of elements of $\text{CAD}(S)$ smaller than or equal to $\mathscr{C}''$. 
    For instance, $\Phi_{46}$ merges the cells surrounding the third section of the fourth cylinder.
    \begin{figure}[h]
        \center 
        \includegraphics[scale=0.2]{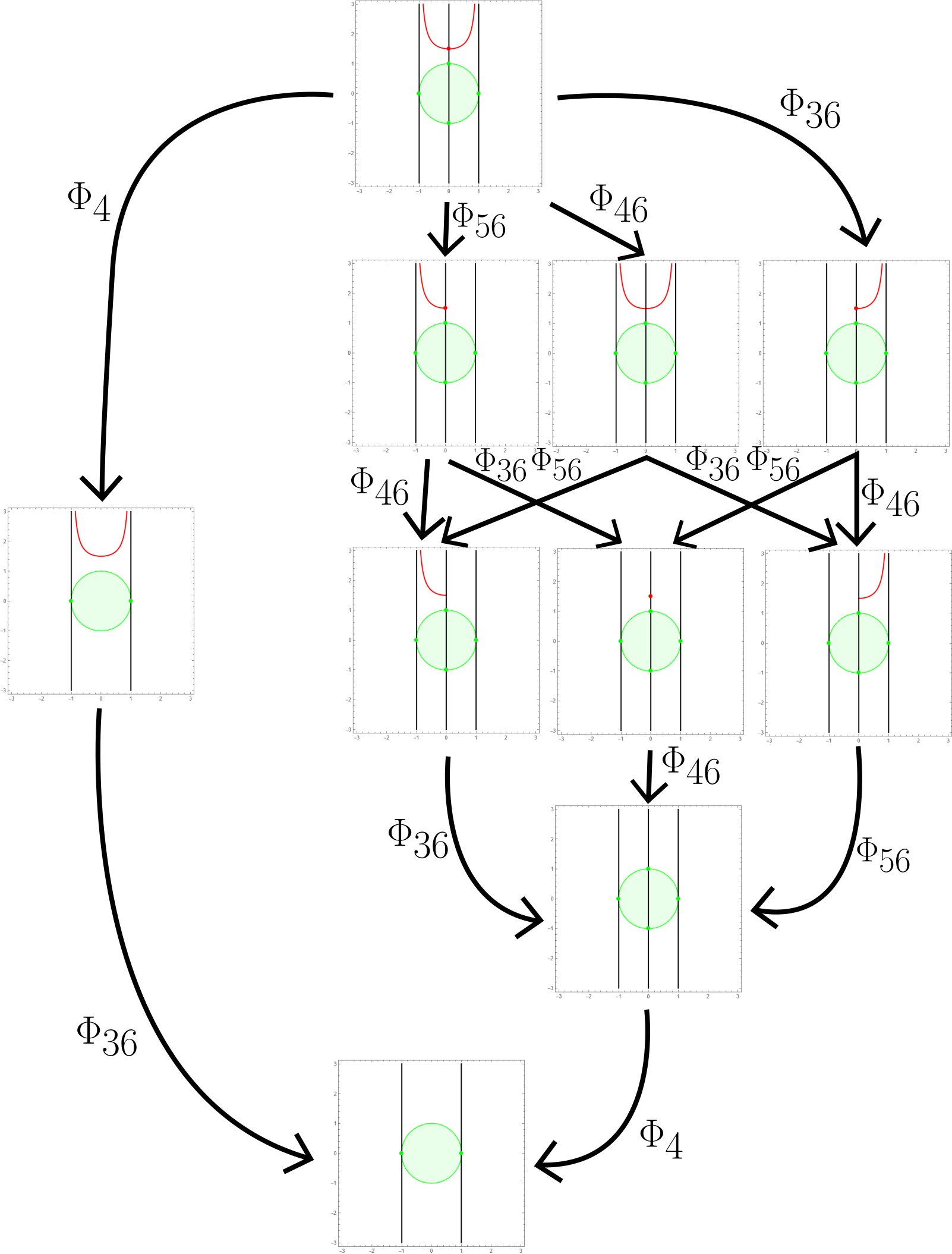}
        \caption{Hasse diagram of elements of $\text{CAD}(S)$ smaller than or equal to $\mathscr{C}''$}
        \label{fig:hasse}
    \end{figure}
\end{example}

\begin{algorithm}[H]
    \caption{\texttt{Minimal}($\mathscr{C}, S$)}
    \label{algo:Min}
\begin{algorithmic}[1]
    \Require $\mathscr{C} \in \text{CAD}(S)$ 
    \Ensure A minimal CAD $\mathcal{M}$ adapted to $S$ such that $\mathcal{M} \preceq \mathscr{C}$ 
    \State Construct the set Red(Tree($\mathscr{C}$)) of reduction rules from Tree($\mathscr{C}$)
    \If{$\exists \Psi \in \text{Red(Tree}(\mathscr{C}))$ s.t. $\Psi$ lifts to a CAD($S$) reduction rule $\Phi$ from $\mathscr{C}$ to $\mathscr{C}'$}
    \State \Return \texttt{Minimal}($\mathscr{C}', S$)
    \Else{}
    \State \Return $\mathscr{C}$
    \EndIf
    \end{algorithmic}
\end{algorithm}

From Theorem \ref{prop:lien-red-ordre}, we deduce immediately Algorithm \ref{algo:Min} which gives an effective way to build a minimal CAD adapted to $S$ as a post-processing operation once an element of $\text{CAD}(S)$ is computed. 
We now apply this algorithm to the CADs $\mathscr{C}$ and $\mathscr{C}'$ adapted to the Trousers $\mathbb{T}$ discussed in Proposition \ref{ex:trousers1}.

\begin{example}\label{ex:trousers2}
    In order to apply  Algorithm \ref{algo:Min} to $(\mathscr{C},\mathbb{T})$ and $(\mathscr{C}',\mathbb{T})$, we compute the sets $\text{Red(Tree}(\mathscr{C}))$ and $\text{Red(Tree}(\mathscr{C}'))$. In this case, we can immediately identify on Figure~\ref{fig:trousers2} the even nodes of these CAD trees satisfying Condition \eqref{eqn:L} and we obtain
    \begin{align*}
        \text{Red(Tree}(\mathscr{C})) = \{\Psi_{12}\} \;\text{ and }\;
        \text{Red(Tree}(\mathscr{C}')) = \{\Psi_{32}\}.
    \end{align*}

    To conclude, it is sufficient to show that neither of these two CAD tree reduction rules lifts to a CAD($\mathbb{T}$) reduction rule. Suppose that $\Psi_{12}$ lifts to a reduction rule from $\mathscr{C}$ to a CAD $\mathcal{D}$ adapted to $\mathbb{T}$. The section $D_{112}$ of $\mathcal{D}$ is the union of cells of $\mathscr{C}$ whose indices are mapped to $112$ by $\psi_{12}$, which are the sections $C_{112}, C_{122}$ and $C_{132}$. In other words, $D_{112}$ is the graph of a continuous semi-algebraic function defined over the union of $C_{11}, C_{12}$ and $C_{13}$.
    This function is obviously not continuous by definition of $\mathbb{T}$. 
    Similar considerations show that $\Psi_{32}$ does not lift to a CAD reduction for $\mathscr{C}'$.
    Hence, the CADs returned by \texttt{Minimal}$(\mathscr{C},T)$ and \texttt{Minimal}$(\mathscr{C}',T)$ are respectively $\mathscr{C}$ and $\mathscr{C}'$. This shows that these CAD are minimal. 
    \begin{figure}[h]
        \center 
        \includegraphics[scale=0.5]{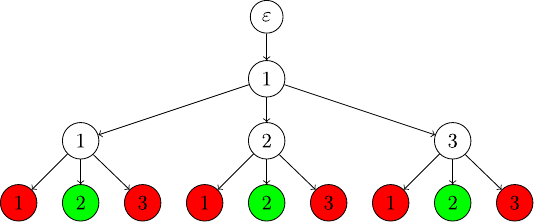}
        \includegraphics[scale=0.5]{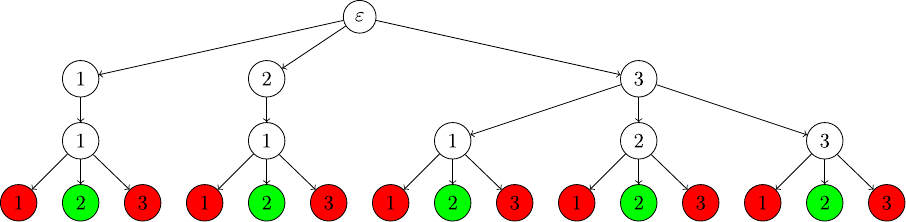}
        \caption{Tree($\mathscr{C}$) and Tree($\mathscr{C}'$)}
        \label{fig:trousers2}
    \end{figure}
\end{example}

\begin{remark}
    In the previous example we just had to show that two particular partitions were not CADs.
    This should be compared with the large number of elements in SSP($\mathscr{C}$) and SSP($\mathscr{C}'$) that would have been processed by direct inspection (see Remark \ref{rem:naive}). 
\end{remark}

\subsection{Minimum CAD and confluence}

In Proposition \ref{ex:trousers1} we proved that $\text{CAD}(\mathbb{T})$ has no minimum by considering two distinct minimal adapted CADs $\mathscr{C}$ and $\mathscr{C}'$. These CADs can actually be obtained by reductions of a common CAD $\overline{\mathscr{C}} \in \text{CAD}(\mathbb{T})$ which is depicted at the top of Figure \ref{fig:trousers3} while $\mathscr{C}$ and $\mathscr{C}'$ are at the bottom. 

\begin{figure}
    \center 
    \includegraphics[scale=0.2]{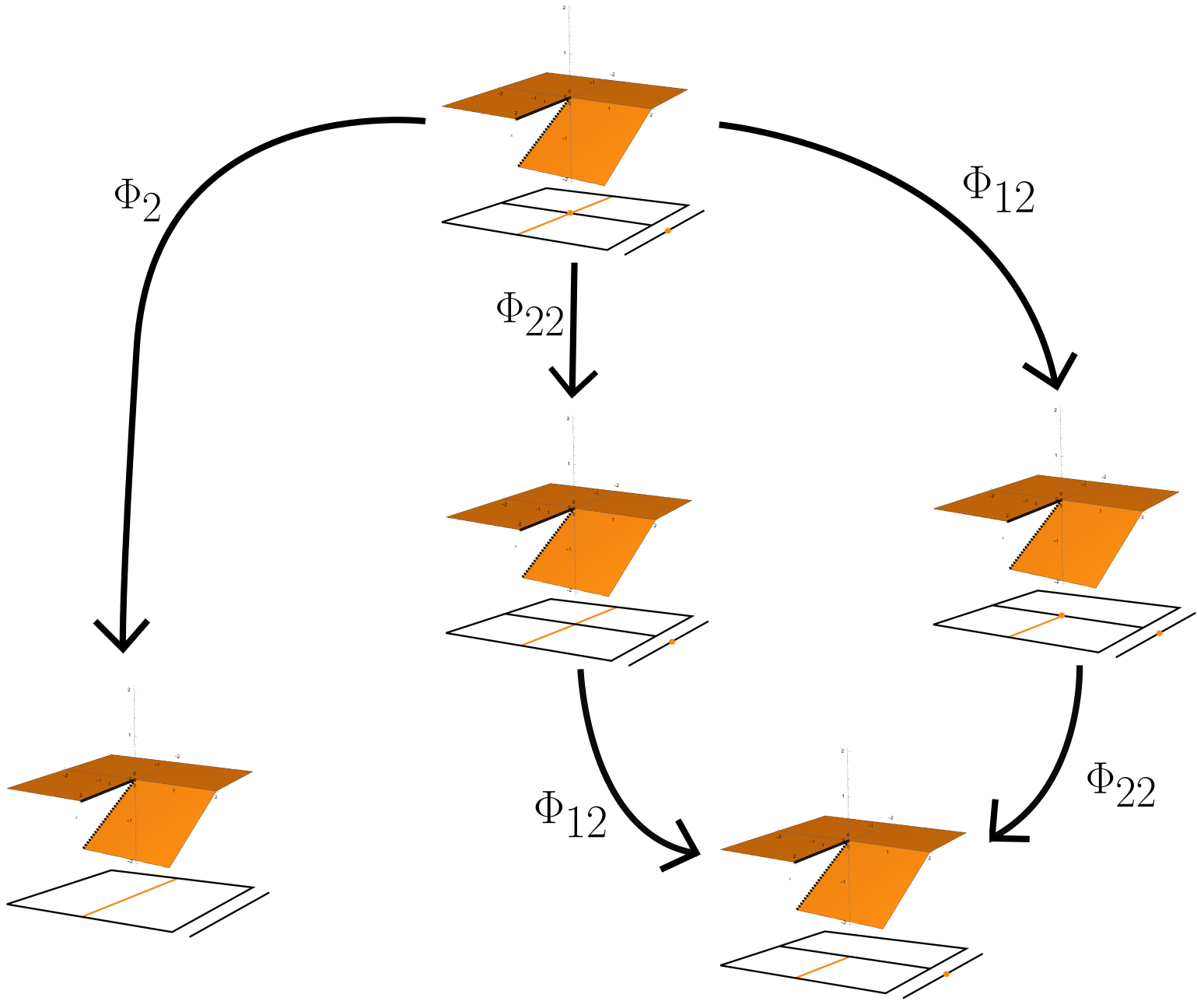}
    \caption{Non-confluence of the trousers}
    \label{fig:trousers3}
\end{figure}

We thus have $\mathscr{C} \stackrel{*}{\leftarrow} \overline{\mathscr{C}} \stackrel{*}{\to} \mathscr{C}'$ but in view of Theorem~\ref{prop:lien-red-ordre}, $\mathscr{C}$ and $\mathscr{C}'$ cannot be reduced to a common CAD. This observation relates the existence of minimum element to the confluence property of $\text{CAD}(\mathbb{T})$. We now show that this is a general phenomenon.

We consider $\text{CAD}(S)$ as the sets of objects of an abstract reduction system (see \cite{TRaAT}) endowed with the reduction relation introduced in Definition \ref{def:CADRed}. Confluence, also known as Church-Rosser or diamond property, appears in various areas of mathematics and computer science. We recall its definition in this particular case.

\begin{definition}
    The reduction system $\text{CAD}(S)$ is globally confluent if for all $\mathscr{C}, \overline{\mathscr{C}}, \mathscr{C}'  \in \text{CAD}(S)$ such that $\mathscr{C} \stackrel{*}{\leftarrow} \overline{\mathscr{C}} \stackrel{*}{\to} \mathscr{C}'$, there exists $\underline{\mathscr{C}} \in \text{CAD}(S)$ such that $\mathscr{C} \stackrel{*}{\to} \underline{\mathscr{C}} \stackrel{*}{\leftarrow} \mathscr{C}'$. It is 
    locally confluent if for all $\mathscr{C}, \overline{\mathscr{C}}, \mathscr{C}'  \in \text{CAD}(S)$ such that $\mathscr{C} \stackrel{}{\leftarrow} \overline{\mathscr{C}} \stackrel{}{\to} \mathscr{C}'$, there exists $\underline{\mathscr{C}} \in \text{CAD}(S) : \mathscr{C} \stackrel{*}{\to} \underline{\mathscr{C}} \stackrel{*}{\leftarrow} \mathscr{C}'$.
\end{definition}
    
Observing that there exists no infinite chain of reductions in $\text{CAD}(S)$ and applying Newman's Lemma (\cite[Lemma 2.7.2]{TRaAT}) we obtain the following result. 
\begin{lemma}
    The rewriting system $\text{CAD}(S)$ is globally confluent if and only if it is locally confluent.
\end{lemma}
In view of this lemma, we say that $\text{CAD}(S)$ is confluent if it is globally confluent or locally confluent.
We can now state and prove the main result of this section which is reminiscent of the characterization of Gröbner bases by means of confluence (see \cite[Chapter~8]{TRaAT}).

\begin{theorem}
    There exists a minimum in $\text{CAD}(S)$ if and only if the reduction system $\text{CAD}(S)$ is confluent.
\end{theorem}
\begin{proof}         
    Assume that $\text{CAD}(S)$ has a minimum $\mathcal{M}$. For every $\mathscr{C},\overline{\mathscr{C}}$ and  $\mathscr{C}' \in \text{CAD}(S)$ such that  $\mathscr{C} \leftarrow \overline{\mathscr{C}} \to \mathscr{C}'$, $\mathcal{M}$ is smaller than or equal to $\mathscr{C}$ and $\mathscr{C}'$. Using Theorem \ref{prop:lien-red-ordre}, this is equivalent to $\mathscr{C} \stackrel{*}{\to} \mathcal{M} \stackrel{*}{\leftarrow} \mathscr{C}'$, so $\text{CAD}(S)$ is confluent.

    We assume now that $\text{CAD}(S)$ is confluent and show that it has a minimum. By Proposition \ref{prop:uniqueMin}, it is sufficient to show that there is at most one minimal element. If $\mathscr{C}$ and $\mathscr{C}'$ are two minimal elements of $\text{CAD}(S)$ we obtain a common finer CAD $\overline{\mathscr{C}}$ of $\mathscr{C}$ and $\mathscr{C}'$ by building a CAD adapted to all the cells of $\mathscr{C}$ and all the cells of $\mathscr{C}'$ (see Remark \ref{rem:Collins}). Hence, we have $\mathscr{C} \stackrel{*}{\leftarrow} \overline{\mathscr{C}} \stackrel{*}{\to} \mathscr{C}.$  By confluence there exists $\underline{\mathscr{C}} \in \text{CAD}(S) : \mathscr{C} \stackrel{*}{\to} \underline{\mathscr{C}} \stackrel{*}{\leftarrow} \mathscr{C}'$. By minimality of $\mathscr{C}$ and $\mathscr{C}'$, we have $\mathscr{C} =\underline{\mathscr{C}}= \mathscr{C}'$, as requested.  
\end{proof}

\begin{acks}
The authors would like to express their gratitude to J. H. Davenport for his advice during the elaboration of this work. They are grateful to the anonymous reviewers for numerous insightful comments. It is a pleasure to thank B. Boigelot, M. England and P. Fontaine for fruitful discussions. 
  P.~Mathonet, L.~Michel and N.~Zenaïdi are supported by the FNRS-DFG PDR Weaves (SMT-ART) grant 40019202.
\end{acks}

\bibliographystyle{ACM-Reference-Format}
\bibliography{references}

\end{document}